\documentclass[a4paper,10pt]{article}

\usepackage{graphicx}
\usepackage{amsmath,amssymb,amsthm}
\usepackage{multicol}

\newtheorem{theorem}{Theorem}[section]
\newtheorem{conjecture}[theorem]{Conjecture}

\newtheorem{lemma}[theorem]{Lemma}
\newtheorem{claim}[theorem]{Claim}

\newtheorem{definition}[theorem]{Definition}

\numberwithin{equation}{section} \numberwithin{theorem}{section}

\newcommand{\prob}{{\mathbb P}}

\newcommand{\Z}{\mathbb Z}

\title{Metastability thresholds for anisotropic bootstrap percolation in three dimensions}

\author{Aernout van Enter$^*$  \and Anne Fey$^\dagger$}

\begin{document}

\maketitle

\begin{abstract}
In this paper we analyze several anisotropic bootstrap percolation models in three dimensions. We present the order of magnitude for the metastability thresholds for a fairly general class of models. In our proofs, we use an adaptation of the technique of dimensional reduction. We find that the order of the metastability threshold is generally determined by the 'easiest growth direction' in the model. In contrast to anisotropic bootstrap percolation in two dimensions, in three dimensions the order of the metastability threshold for anisotropic bootstrap percolation can be equal to that of isotropic bootstrap percolation. 
\end{abstract}


\bigskip

{\small \emph{$^*$ Aernout van Enter, Johann Bernoulli Institute for Mathematics and Computer science, Groningen University, The Netherlands, \texttt{aenter@phys.rug.nl}}}

{\small \emph{$^{\dagger}$ Anne Fey, Delft Institute of Applied Mathematics, Delft University of Technology, The Netherlands, \texttt{a.c.fey-denboer@tudelft.nl}}}

\bigskip

\section{Introduction}

In bootstrap percolation models on a finite cube $[0,L]^d \in \Z^d$, in the starting configuration every site is occupied with probability $p$, and empty otherwise, independent of all other sites. 

The configuration then evolves according to the bootstrap rule: each site that 
has at least $k$ occupied sites in its neighborhood becomes occupied, and 
occupied sites remain occupied. The rule will repeatedly be applied until no new 
site will become occupied. One is usually interested in the probability that the cube is internally spanned, that is, in the final configuration all sites are occupied. 
We will in this paper choose $k$ to be half the number of 
sites in the neighborhood.
In ordinary bootstrap percolation  
the neighborhood of a site consists of all its nearest neighbors. In anisotropic
 bootstrap percolation, however, the size or shape of the neighborhood is not 
equal in every direction. For earlier work on bootstrap percolation models, see 
e.g \cite{ADE, AL, Dua, vE, Grifp, Holp, HLR, Mou, schonmann, Sch2}.

Bootstrap percolation models and arguments have been applied in a 
variety of settings, from fluid dynamics, 
magnetic models, the theory of glasses,  
neural networks, the theory of sandpiles to rigidity theory and economics, see 
e.g. \cite{Adl,Ami, CRV, ET, FLP, LV, L, Ton}.

Anisotropic bootstrap percolation models until now have been studied in two 
dimensions. For instance, Gravner and Griffeath \cite{GG} introduced the model where the 
neighborhood consists of six sites, namely, in the $x$ direction only the 
nearest neighbors, but in the $y$ direction both the nearest and the 
next-nearest neighbors. We will call this the $(1,2)$ model. In this model (where $k=3$), an occupied rectangle can grow by the bootstrap rule most 
easily in the $y$ direction. A single occupied site at distance 1 or 2 from 
the square suffices to fill the next line segment in this direction, whereas in the $x$ direction, two occupied sites with no more than three empty sites in between are needed to fill the next line segment. 
The behaviour of this model is similar to that of the semioriented
model studied in \cite{Dua, Mou, Sch2}; for a  bootstrap rule whose 
anisotropy appears to be of qualitatively 
different type, as its asymptotics follows standard isotropic behaviour, see 
\cite{BMM}.

We will study the following kind of anisotropic model (see Figure \ref{neighborhoodfig}): the neighborhood consists of the $a$-nearest neighbors in the $x$ direction, the $b$-nearest neighbors in the $y$ direction and the $c$-nearest neighbors in the $z$ direction, with $a \leq b \leq c$, and we choose $k = a+b+c$. 
We will call this the $(a,b,c)$ model. The notation $(a,b)$ for two-dimensional models, like the $(1,2)$ model of \cite{GG}, is similar.

\begin{figure}%
\centering
\includegraphics[width=6cm]{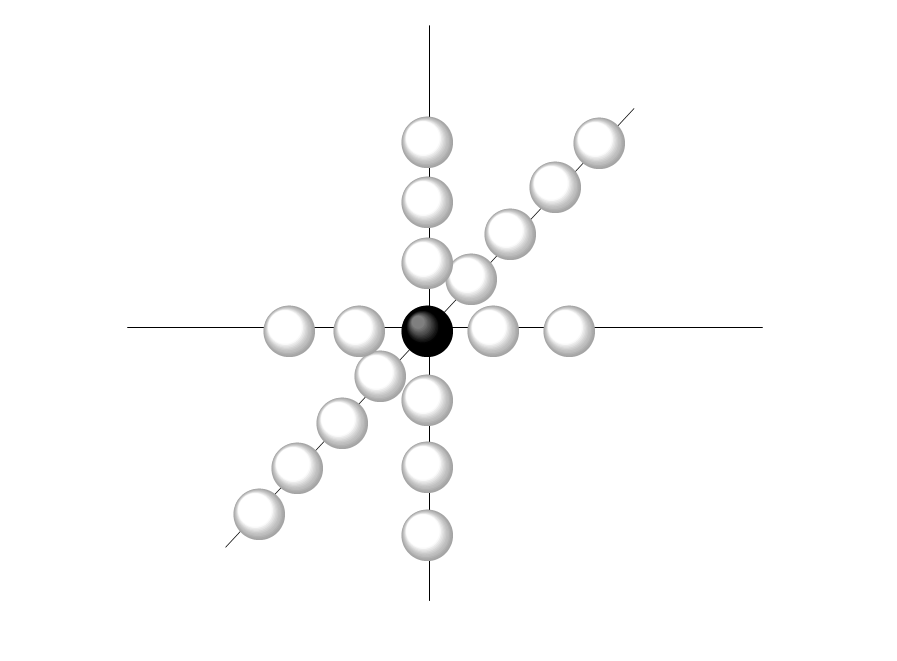}%
\caption{Illustration of the neighborhood of a site in the $(2,3,4)$ model. The $z$ axis is towards the reader.}%
\label{neighborhoodfig}%
\end{figure}

An important tool that we will use is dimensional reduction \cite{schonmann}. 
Suppose a certain large threedimensional rectangular block is occupied. Then all sites in a slab (we call the set of sites outside the block that are adjacent to one of the faces, a slab), already have some number $k' = a$, $b$ or $c$ of occupied sites in the intersection of their neighborhood with the rectangular block. Thus, once the rectangular block is occupied, a slab in direction $x$ will become occupied in the $(a,b,c)$ model, if the slab is (internally) spanned in the $(b,c)$ model. For the $(a,b,c)$ model, we will call the $(b,c)$ model the {\em reduced model} in the $x$ direction, and likewise define the reduced models in the $y$ and $z$ directions. 

Our main example is the $(1,1,2)$ model. The $(1,1,2)$ model has two hard growth directions and one easy one. Namely, in the $z$ direction the reduced model is $(1,1)$, but for the other directions it is $(1,2)$. Our main result however is valid for all $(a,b,c)$ models.

\section{Terminology}

In this section we present some common terms in bootstrap percolation, which we will use throughout the paper in the sense explained below.

\paragraph{Threshold lengths, sharp thresholds}
Call $\prob([0,L]^d \mbox{ internally spanned})$ the probability that for a bootstrap percolation model on a finite cube $[0,L]^d$, in the final configuration all sites are occupied. 
The typical behavior of bootstrap percolation models is that there is a sharp threshold length, this means that 
there is a function $f(\frac{1}{p})$ such that, as $p \to 0$,
\begin{eqnarray}
\nonumber \prob([0,L]^d \mbox{ i.s.}) \to 0 & \mbox{ for } & L < e^{(1- \varepsilon)f(\frac{1}{p})} \ (\mbox{in} \ d=2), \ \mbox{ or} \\
\nonumber \prob([0,L]^d \mbox{ i.s.}) \to 0 & \mbox{ for } & L < e^{e^{(1- \varepsilon)f(\frac{1}{p})}} \ (\mbox{in} \ d=3), \  \mbox{ and} \\ 
\nonumber \prob([0,L]^d \mbox{ i.s.}) \to 1 & \mbox{ for } & L > e^{(1+ \varepsilon)f(\frac{1}{p})} \ (\mbox{in} \ d=2), \ \mbox{ or} \\
\nonumber \prob([0,L]^d \mbox{ i.s.}) \to 1 & \mbox{ for } & L > 
e^{e^{(1+ \varepsilon)f(\frac{1}{p})}} \ (\mbox{in} \ d=3). \\
\label{threshold}
\end{eqnarray}

We will call $e^{f(\frac{1}{p})}$ (in $d=2$), or $e^{e^{f(\frac{1}{p})}}$ (in $d=3)$ the threshold length $L^{th}(p)$.

Inversely, there is a percolation theshold $p_{th}(L)= Cg(L)$, such that asymptotically for increasing $L$ for all positive $\varepsilon$ there will be percolation if $p > (1+\varepsilon)g(L)$ and no percolation if $p < (1- \varepsilon)g(L)$, with high probability.  The first proof of a sharp threshold length was for the isotropic model in dimension 2 by Holroyd \cite{holroyd}. Since then, several more sharp threshold results have been obtained \cite{BBDM,BBM,DE,DH,holroydmodified, HLR}.

\paragraph{Order of threshold length}
There are many bootstrap percolation models for which there is no proof (yet) of a sharp threshold length. For some of those 
less precise results
known, namely, lower and upper bounds for $L^{th}(p)$ similar to those of \eqref{threshold}, but with different multiplicative constants $\gamma$ and $\Gamma>\gamma$, namely:
\begin{eqnarray}
\nonumber \prob([0,L]^d \mbox{ i.s.}) \to 0 & \mbox{ for } & L < L^-(p) = e^{\gamma f(\frac{1}{p})} \ (\mbox{in} \ d=2), \ \mbox{ or} \\
\nonumber \prob([0,L]^d \mbox{ i.s.}) \to 0 & \mbox{ for } & L < L^-(p) = e^{e^{\gamma f(\frac{1}{p})}} \ (\mbox{in} \ d=3), \  \mbox{ and} \\ 
\nonumber \prob([0,L]^d \mbox{ i.s.}) \to 1 & \mbox{ for } & L > L^+(p) = e^{\Gamma f(\frac{1}{p})} \ (\mbox{in} \ d=2), \ \mbox{ or} \\
\nonumber \prob([0,L]^d \mbox{ i.s.}) \to 1 & \mbox{ for } & L > L^+(p) = 
e^{e^{\Gamma f(\frac{1}{p})}} \ (\mbox{in} \ d=3). \\
\label{thresholdorder}
\end{eqnarray}

In such a case we say that the order of the threshold length is known. 

In this paper, we focus on $f(\frac{1}{p})$, that is, orders of the threshold length of the following kind:
\begin{itemize}
	\item For $(a,b)$ models, we are interested in the order of $\ln L^{th}(p)$,
	\item For $(a,b,c)$ models, we are interested in the order of $\ln \ln L^{th}(p)$. 
\end{itemize}
Indeed, our main result is the determination of the order of $\ln \ln L^{th}(p)$ for the general $(a,b,c)$ model. 
However, since strictly speaking we do not prove that a sharp threshold length exists, we will not use this term in our theorem.

In the remainder of this paper, we will use subscripts $a,b$ and $a,b,c$ to refer to the $(a,b)$ model resp. the $(a,b,c)$ model.
Sometimes we will use the term $L^{th}_{a,b}(p)$ in cases where no sharp threshold result is known, that is, $f_{a,b}(1/p)$ is known only up to a constant. We will only use this abusive notation in cases where the constant is unimportant, to avoid cumbersome elaborations in terms of the lower and upper bounds. 

\smallskip

We conjecture that in fact for all the models we consider, that is, every $(a,b)$ model in two dimensions, and every $(a,b,c)$ model in three dimensions, there exists a sharp threshold length.

\paragraph{Critical droplets}

Often, the proof for an upper bound $L^+(p)$ for $L^{th}(p)$ involves the notion of a ``critical droplet''. This is an occupied connected set of sites of a size and shape such that, if all other sites in a finite or infinite lattice are independently occupied with probability $p$, then the droplet will continue to grow with high probability. Since the occurrence of such a critical droplet is (for fixed $p$) a local event, if $L$ is large enough, that is to say $L$ increases fast enough as $p \to 0$, then there will with high probability be a critical droplet in the volume of linear size $L$.

As a simple example, consider isotropic bootstrap percolation in one 
dimension, that is, on the interval $[0,L]$. Suppose that $L$ increases faster 
than $1/p$ as $p \to 0$. In this case, the critical droplet is one occupied 
site. The probability that somewhere on the line there is an 
occupied site, tends to 1. Also, this critical droplet grows with probability 1: once there is at least one occupied site, then 
the whole line is spanned. 

Note that there are multiple possibilities for choosing a critical droplet. The smaller the size of the critical droplet one chooses, the tighter upper bound for $L^{th}(p)$ one obtains. In this paper, our choice is always a certain occupied rectangle (rectangular block). 
(We expect that this choice of rectangles is not optimal; in studies such as \cite{holroyd} for example, the critical droplet is more subtle: it is a certain set of occupied sites such that a rectangle of size $O(1/p)$ is internally spanned rather than fully occupied.) 

Suppose the number of occupied sites in the critical droplet is $V$. Then the probability $P$ that a fixed site is in a critical droplet, is of order $p^V$. Then, disregarding some corrections due to the finite size of the droplet (which are negligible if $L >> V$, and $p$ small):
\[
\prob([0,L]^d \mbox{ i.s.}) \geq 1-(1-P)^{L^d} \approx 1 - e^{-PL^d},
\]
where we used that $(1-P)^{\frac{1}{P}}$ tends to $e^{-1}$ as $p \to 0$.
We can alternatively interpret this formula as follows: suppose we have a Poisson point process with parameter $P$. Then the above probability is the probability that there is at least one Poisson point in the cube $[0,L]^d$. For this reason, one also calls $P$ the density of critical droplets. 
 We see from this expression that $\prob([0,L]^d \mbox{ i.s.})$ tends to 1 if $L^d$ tends to infinity faster than $1/P$. We conclude that 
for every $P$ denoting the density of a certain choice of critical droplets, 
\begin{equation}
L^+(p) = O(P^{-1/d}).
\label{density+}
\end{equation}
Inversely, we can say that if $L<L^-(p)$, then with high probability there is no critical droplet of any kind, so that for every choice of critical droplet, we have
\begin{equation}
P = O(L^-(p))^{-d}.
\label{density-}
\end{equation}

\paragraph{Supercritical size}
We say that a volume is of supercritical size if the probability that it is internally spanned, tends to 1 as $p \to 0$. For example, the volume $[0,L]^d$, with $L>L^{+}(p)$, is of supercritical size. However, a volume of supercritical size does not need to be cubic. We will also consider rectangular volumes of supercritical size. For example, a rectangle is of supercritical size if it consists of at least $(L^{+}(p))^d$ sites, and its shape is such that it is much larger than the critical droplet, in every direction. 

\section{Main result}

Our main result is that the order of $\ln \ln L^{th}(p)$ for the general $(a,b,c)$ model, depends only on $a$ and $b$. 
We state this result as follows:

\begin{theorem}
For the $(a,b,c)$ model, 
there exist constants $\gamma_{a,b,c}$ and $\Gamma_{a,b,c}$ such that, as $p \to 0$, 
\begin{eqnarray*}
\nonumber \prob([0,L]^3 \mbox{ i.s.}) \to 0 & \mbox{ for } L < e^{e^{\gamma_{a,b,c}f_{a,b}(1/p)}},\\
\nonumber \prob([0,L]^3 \mbox{ i.s.}) \to 1 & \mbox{ for } L > e^{e^{\Gamma_{a,b,c}f_{a,b}(1/p)}}.\\
\end{eqnarray*}
where 
\begin{itemize}
\item if $a=b$ then  $f_{a,a}(1/p) = p^{-a} + o(p^{-a})$,
\item of $a<b$ then  $f_{a,b}(1/p) = p^{-a}\ln^2 p + o(p^{-a} \ln^2 p)$.
\end{itemize}
\label{maintheorem}
\end{theorem}

Since the proof of this theorem involves dimensional reduction, we use information on $f_{a,b}(1/p)$, the order of the threshold length of the $(a,b)$ model. In several cases a sharp threshold length is known, namely, for the $(1,1)$ model \cite{holroyd} and for the $(1,b)$ model with $b>1$ \cite{DE}. In these cases we can specify the constant $\Gamma_{a,b,c}$. For example, it will turn out that if $a=b=1$, then the constant $\Gamma_{1,1,c}$ can be chosen to be twice the Holroyd constant $2C_H = \frac{\pi^2}{9}$. 

Furthermore, from the analysis of Duminil-Copin and Holroyd \cite{DC,DH}, it follows that for the $(a,a)$-model a sharp threshold result holds, namely that there is a constant $C$ such that \eqref{threshold} holds with $f_{a,b}(\frac{1}{p}) = e^{Cp^{-a}+o(p^{-a})}$. We summarize the current knowledge in Table \ref{knownorderstable}.

\medskip

\begin{table}
\centering
\begin{tabular}{ll}
$(a,b)$ & order of $\ln L^{th}_{a,b}(p)$ \\  [1ex] \hline \\ [-1.5ex]
$(1,1)$ & ${\frac{1}{p}}$ \cite{holroyd} \\
$(1,b)$ with $b>1$ & ${\frac{1}p \ln^2\frac 1p}$  \cite{DE}\\
$(a,a)$ & ${\frac{1}{p^a}}$ first mentioned in \cite{GG} (Section 7), see also \cite{DC,DH}\\
\end{tabular}
\caption{All known sharp thresholds for $(a,b)$ models.} 
\label{knownorderstable}
\end{table}

For the general case of the $(a,b)$ model with $1<a<b$ it is not known whether a sharp threshold exists. For the $(a,b)$ model with $a<b$, we derive the following result (see Section \ref{generallowsection}), which is sufficient for our purposes:

\begin{claim}~
Let $a <b$, then for the $(a,b)$ model, there exist constants $\gamma_{a,b}$ and $\Gamma_{a,b}$, such that
\begin{eqnarray}
\nonumber \prob([0,L]^2 \mbox{ i.s.}) \to 0 & \mbox{ for } & L < L^-_{a,b}(p) = e^{\gamma_{a,b} p^{-a}\ln^2 p}, \\
\nonumber \prob([0,L]^2 \mbox{ i.s.}) \to 1 & \mbox{ for } & L > L^+_{a,b}(p) = e^{\Gamma_{a,b} p^{-a}\ln^2 p}.\\
\end{eqnarray}
\label{claim}
\end{claim}

{\bf Remark.} 
For the $(1,1,1)$ model, a sharp threshold result is already known \cite{BBM}. One might expect that this result should give us a lower bound for $\gamma_{a,b,c}$, since it seems natural to expect (and indeed we do) that an $(a,b,c)$ model with a larger neighborhood should have more difficulty in growing, and in filling up a volume. 
However, a direct inequality is not that obvious. Indeed, as a pair of occupied sites on a line parallel to the $y$-axis at distance 4 in the $(1,2)$-model can cooperate, whereas they cannot in the $(1,1)$-model (as then they will never belong to the neighborhood of the same site) a configuration-wise ordering is excluded. 
 
Thus, we present our intuition as a series of conjectures: 

\begin{conjecture}~
\begin{itemize}
\item
If $a<b$, then there is a constant $C_{a,b}$ such that 
\[
\ln L_{a,b}^{th}(p) = C_{a,b}p^{-a} \ln^2(1/p) + o(1/p^{-a} \ln^2(1/p)),
\]
\item If $b'>b$, then $C_{a,b'} > C_{a,b}$.
\end{itemize}
\label{nextproject}
\end{conjecture}

\begin{conjecture}~
\begin{itemize}
\item 
There is a constant $C_{a,b,c}$ such that 
\[
\ln\ln L_{a,b,c}^{th}(p) = C_{a,b,c} f_{a,b}(1/p) + o(f_{a,b}(1/p)),
\]
where $f_{a,b}(1/p)$ is as in Theorem \ref{maintheorem}.
\item If $a'\geq a$, $b'>b$, and $c'\geq c$ then $\ln\ln L_{a',b',c'}^{th}(p) \geq \ln\ln L_{a,b,c}^{th}(p)$.
\end{itemize}
\end{conjecture}

We expect that the first half of Conjecture \ref{nextproject} might be derived by extending the techniques from \cite{DE}. Indeed, by \cite{DE}, the conjecture holds for $a=1$. Nevertheness, the full proof is outside the scope of the present paper.

\section{The lower bound}

We first explain the lower bound in the case $(a,b,c) = (1,1,c)$ in some detail, as an instructive example. 

\subsection{The lower bound for the $(1,1,c)$ model}

In this section, we show that for all $L(p)$ larger than the value stated in Theorem \ref{maintheorem}, the probability that the configuration is internally spanned, tends to 1.
We choose an occupied rectangular block as critical droplet. This rectangular block will keep growing if the areas of its faces are large enough. For instance, if the area of the face in the $z$ direction is larger than $e^{2C_H(1+\varepsilon)/p}$, then the slab next to it is spanned with high probability, because its size is supercritical for the $(1,1)$ model. 

\begin{figure}%
\centering
\includegraphics[width=8cm]{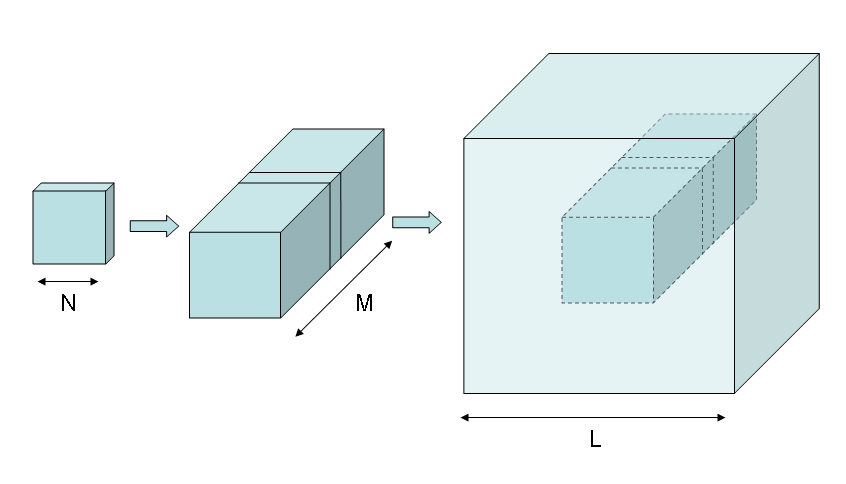}%
\caption{Illustration of how the critical droplet grows in the $(1,1,c)$ model (not to scale). The $z$ direction is towards the reader. First the droplet grows only in the $z$ direction, then in all directions.}%
\label{112plaatje}%
\end{figure}

Naively, one could suppose that we need a critical droplet with one face of supercritical size for the $(1,1)$ model, and two of supercritical size for the $(1,c)$ model. However, we can do better than that, following essentially the same line of thought as in \cite{EH} which is similar to that in \cite{ADE} and which was based on an unpublished observation of Roberto Schonmann. There a suitable critical droplet for the $(1,2)$ model was found to be a strip of length $\frac {C'}p \ln \frac 1p$, and width 2, with $C'$ a large enough constant. Clearly the width of this strip is not supercritical, but the strip grows with large probability into a rectangle of size $\frac Cp \ln \frac 1p$ times $\frac 1{p^2}$, which is of supercritical size in both directions. The probability that such a critical droplet is at a fixed position, is $p^{\frac {C'}p \ln \frac 1p} = e^{-\frac {C'}p \ln^2 \frac 1p}$. 

For the $(1,1,c)$ model, we choose as critical droplet a rectangular block of size $N$ by $N$ by 2, with $N$ such that the droplet will grow sufficiently far in the $z$ direction. More specifically, let us define
\begin{equation}
P_{1,1} = e^{-\frac{2C_H}{p}},  
\label{specificP}
\end{equation}
which is the probability that a fixed site is 
in a critical droplet in the $(1,1)$ model. Then a square of size $N^2 \geq e^{\frac{2C_H(1+\varepsilon)}{p}}$ is supercritical for the $(1,1)$ model, because with high probability there will be a critical droplet in it.

We will choose $N^2=(P_{1,1})^{-(1+\epsilon)}$ slightly larger than minimally supercritical, because we need that many slabs of that size will get occupied with high probability, rather than just one. For our critical droplet to grow into a rectangular block with all faces of supercritical size, we need that $M$ adjacent slabs get occupied with large probability, with $M$ such that $MN \geq (L^{th}_{1,c}(p))^2$. As was proved in \cite{EH} (and explained above), there is a constant $\Gamma_{1,2}$ such that $L^{+}_{1,2}(p) \leq e^{\frac {\Gamma_{1,2}}p \ln^2 \frac 1p}$. In \cite{DE}, this result was extended: a sharp threshold length of the form $L^{th}_{1,c}(p) = e^{\frac {C(c)}p \ln^2 \frac 1p}$ was found for all $(1,c)$ models, where $C(c)$ is a constant depending only on $c$.  
The probability that a slab contains a critical droplet for the $(1,1)$ model is at least $1-(1-P_{1,1})^{N^2}$, and therefore the probability that $M$ slabs will get occupied, is at least $(1-(1-P_{1,1})^{N^2})^M$ (this all up to some irrelevant constants, due to the finite size of the droplets, possible overlapping of the droplets, etc). 

The logarithm of this probability is approximately \\
$-M e^{-P_{1,1} N^2} = -M e^{-(P_{1,1})^{-\varepsilon}} = -Me^{-e^{\frac {2C_H\varepsilon}{p}}}$, which is close to 0 for all $M = o(e^{e^{\frac {2C_H\varepsilon}{p}}})$. We see that it is possible to choose $M$ large enough so that $MN \geq (L^{th}_{1,c}(p))^2$ is satisfied.

Now we need to find a bound for $L$ such that a cube $[0,L]^3$ contains a critical droplet with high probability. This is the case if $L^3 \geq p^{-2N^2}$ (see \eqref{density+}).
We work out
\[
L \geq p^{-\frac 23 N^2} = e^{\frac 23 \ln \frac1p e^{\frac{2C_H(1+\varepsilon)}{p}}} = e^{e^{\frac{2C_H(1+\varepsilon)}{p} + o(\frac 1p)}},
\]
which gives us the desired lower bound. 

We remark that it is possible to refine this argument and obtain a smaller estimate for $N$. For instance, if we suppose that $M = N^k$ for some $k$, then we find that an $N$ by $N$ by 2 occupied rectangular block with $N^2 = \frac {k}{P_{1,1}} \ln \frac 1{P_{1,1}}$ is sufficiently large to act as a critical droplet. However, this refinement leads to the same conclusion.

\subsection{The lower bound in the general case}
\label{generallowsection}

We start with deriving Claim \ref{claim}, that is, we derive some basic estimates for $L^{-}_{a,b}(p)$ and $L^{+}_{a,b}(p)$.
Since the proof mostly consists of repeating work that has already been published, here we just sketch the steps that are needed. Namely, considering that

\begin{itemize}
	\item an occupied rectangle of size at least $C p^{-a} \times b$, with $C$ large enough, serves as a critical droplet,
as such an occupied rectangle 
grows with high probability into an occupied rectangle of size $C p^{-a}\ln \frac 1p \times p^{-b}$, in the same manner as the 2 by $C \frac 1p \ln \frac 1p$ critical droplet of \cite{EH} (the generalization is straightforward), 
	\item the probability of a fixed $2 \times C p^{-a}\ln \frac 1p$ rectangle to be occupied, is $p^{C p^{-a}\ln \frac 1p} = e^{-C p^{-a} \ln^2 \frac 1p}$,
	\item therefore, 
$L^{+}_{a,b}(p) \leq  e^{C p^{-a} \ln^2 \frac 1p}$,
\end{itemize}
and we obtained one of the inequalities.

For the bound in the other direction, again a generalization of the method of \cite{EH} for the $(1,2)$ model is needed. First, we generalize (2) of \cite{EH} to \eqref{anisospanned} (this paper, Section \ref{defsandlemmas}). Then we repeat the calculation that leads to (5) of \cite{EH}, replacing $p^2$ by $\tilde{p}^b$ and $p$ by $\hat{p}^a$, and choosing $x = C_2 \frac 1{p^a}\ln \frac 1p$ and $y = p^{-b+1/2}$ instead of choosing $k=\frac 1{p^{3/2}}$ and $l = C_2 \frac 1p \ln \frac 1p$. This will lead to the other inequality.

\medskip
Now we get back to the $(a,b,c)$ model in three dimensions. 
We will show that for all $L$ larger than the value stated in Theorem \ref{maintheorem}, the probability that the configuration is internally spanned, tends to 1. Suppose that there is an occupied rectangular block in the cube $[0,L]^3$, then for each direction, the size of its face will determine whether it will grow easily in that direction (also, the width of the rectangular block should be large enough; width $c$ is sufficient). For an occupied rectangular block to serve as a critical droplet, at least one of its faces should be of supercritical size for the reduced model in that direction. On the other hand, if all of the faces of an occupied rectangle are of supercritical size, then it is certainly a critical droplet.

We suppose for ease of notation that the easiest growth direction is the $z$ direction. Only if $a<b$ and Conjecture \ref{nextproject} does not hold, then it is possible that the $(a,c)$ model has a lower threshold length than the $(a,b)$ model. But in that case, these two models do have the same order of threshold length, therefore, in that case the $y$ and $z$ direction would behave essentially in the same manner. 

We denote 
\[
\bar{L}^+(p) = \max\{L^+_{a,b}(p),L^+_{a,c}(p),L^+_{b,c}(p)\}.
\] 

Let $P_{a,b}$ denote the density of critical droplets for the $(a,b)$ model (see \eqref{density+}).

We choose as a critical droplet for the $(a,b,c)$ model an occupied $N$ by $N$ by 2 rectangular block, where $N \geq (L_{a,b}^{+}(p))^{1+\varepsilon}$.
This rectangular block will grow easily into a $N$ by $N$ by $M$ rectangular block that is of supercritical size in all directions, that is, $MN \geq (\bar{L}^{+}(p))^2$.

The probability that $M$ slabs will get occupied, is at least $(1-(1-P_{a,b})^{N^2})^M$ (up to some irrelevant constants). 
We take the logarithm, which is approximately $-M e^{-P_{a,b} N^2} = -M e^{-P_{a,b}^{-\varepsilon}}$. This is close to 0 if $M = o(e^{(L_{a,b}^-(p))^{\epsilon}})$.
By the bounds in Claim \ref{claim}, we see that it is possible to satisfy both $M \geq \bar{L}^+(p)$ and $M = o(e^{(L_{a,b}^-(p))^{\epsilon}})$.

The probability that a fixed $N$ by $N$ by 2 rectangular block is occupied, is $p^{2N^2}$. Therefore, $p^{2N^2}$ describes the density of critical droplets in the $(a,b,c)$ model, and we obtain (see \eqref{density+})
\[
L^3 \geq (1/p)^{2N^2}.
\]

We work out
\begin{equation}
L \geq (1/p)^{\frac 23(L^+(p))^{2+2\varepsilon}} = e^{e^{Cf_{a,b}(1/p)(1+\varepsilon) + o(f_{a,b}(1/p))}},
\label{generalL}
\end{equation}

where $C$ is a suitable constant. This gives the lower bound for $L$ in Theorem \ref{maintheorem}. 

\section{The upper bound}

In this section, we show that for all $L$ smaller than the value $L_{a,b,c}^-(p)$ stated in Theorem \ref{maintheorem}, 
the probability that the configuration is internally spanned, tends to 0.
In the case that $L \leq (1/p)^{1/4}$, we simply estimate $\prob([0,L]^3 \mbox{ i.s.}) \leq p(p^{-1/4})^3 = p^{1/4}$, which tends to 0 if $p \to 0$. In the case that $L>(1/p)^{1/4}$ however, we need to do a lot more work.

We will apply the method introduced in \cite{cerf} for the case of threedimensional isotropic bootstrap percolation, later generalized to arbitrary dimension in \cite{cerf2}, and again presented in \cite{holroydmodified}, where some minor inaccuracies in the original presentation are corrected.
We modify the method to be applicable in the case of anisotropic bootstrap percolation. 

A necessary condition for a cube $[0,L]^3$ to be internally spanned in the $(a,b,c)$ model is that it is internally {\em crossed}, that is, in the final set of occupied sites, there is a path of occupied sites connecting two opposite faces of the cube, in every direction. To find a suitable bound for the probability of $[0,L]^3$ to be internally crossed, we will enhance the configuration, that is, occupy more sites. 
Throughout the remainder, we will suppose for ease of notation that the $z$ direction is the easiest growth direction. Again, even though if Conjecture \ref{nextproject} does not hold then the easiest growth direction might be the $y$ direction, this possible exception does not change our proof, since in that case the $y$ and $z$ direction behave essentially in the same manner.

\subsection{Some definitions and lemmas}
\label{defsandlemmas}

\begin{definition}
In {\em weakly enhancing} the initial configuration in a connected set of sites $\rho$, where $R = [0,r_x]\times[0,r_y]\times[0,r_z]$ is the smallest rectangular block covering $\rho$, we do the following:

\begin{enumerate}
\item we evolve according to the $(a,b,c)$ bootstrap rule, only taking into account occupied sites in $\rho$. 
\item We occupy sites that are {\em weakly spanned}. We find weakly spanned sites as follows: for every site that is occupied, but was initially empty, we draw lines from this site to every site in its neighborhood that is occupied in the final configuration. A weakly spanned site is an empty site that is crossed by (at least one) such a line. We give an example in Figure \ref{weaklyspannedfig}. 
\end{enumerate}

We say a rectangular block is {\em weakly crossed} if it is internally crossed in every direction after weakly enhancing.
\end{definition}

\begin{figure}%
\centering
\includegraphics[width=9cm]{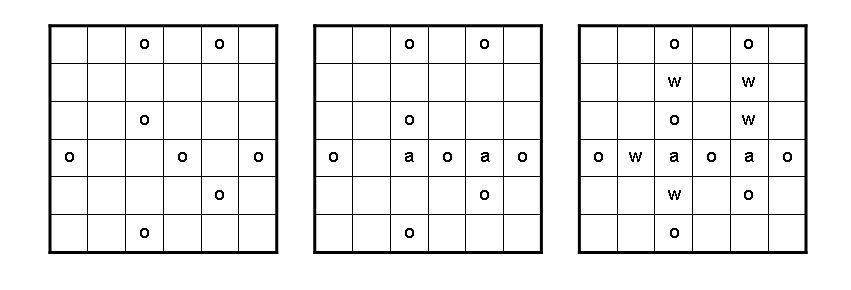}%
\caption{An example to explain weakly spanned sites. Consider the $(2,3)$ model, in a 6 by 6 rectangle. On the left is the initial configuration (occupied sites are marked o), in the middle the final configuration (spanned sites are marked s), and on the right we have added weakly spanned sites (marked w). We see that the rectangle is weakly crossed.}%
\label{weaklyspannedfig}%
\end{figure}

The following lemma is somewhat similar to  Lemma 1 from \cite{AL}. Related arguments continue to reappear in various forms in many papers on bootstrap percolation. 

\begin{lemma}

For the $(a,b,c)$ model, suppose that the cube $[0,l]^3$ is internally spanned. Then there are constants $\kappa$ and $\lambda$ that only depend on $a$, $b$ and $c$, such that for every $k \in [1,\frac{l-\lambda}{\kappa}]$, there is a rectangular block with longest side in $[k, \kappa k + \lambda]$ that is weakly crossed. 
\label{spannedsubsets}
\end{lemma}

\begin{proof}
First, we need some further notation. We say a {\em region} is a set of occupied sites such that the set is connected after weakly enhancing it. Therefore, if $\rho$ is a region then the smallest rectangular block that covers it, is weakly crossed. 
We say that an empty site $x$ has a {\em set of influencing regions} if there is a set of regions such that if we empty the entire configuration except these regions then $x$ is spanned, but if we furthermore empty any one of these regions then $x$ is not spanned. There may be several possibilities for the set of influencing regions. The cardinality of these sets is always between 1 and $a+b+c$, and the distance between any two regions in one of these sets is at most $2c+1$. We say that the influencing rectangular block is the smallest rectangular block such that it covers the set of influencing regions.

Next, we present a special order of occupying sites according to the bootstrap rule. After step $i$, we will have a set of regions $\mathcal{C}_i$. We start with the set $\mathcal{C}_1$ of regions with diameter 1, that consists of all single occupied sites in the initial configuration. Since we suppose that $[0,l]^3$ is internally spanned, $\mathcal{C}_1$ is not empty. Then we iterate in the following manner:

\begin{itemize}
\item If there are two regions such that their union is again a region, then $\mathcal{C}_{i+1}$ consists of all the regions in $\mathcal{C}_i$ with these two regions removed, and with a new region added that is the union of these two regions. If there are several such pairs of regions, then we make an arbitrary choice.
\item If not, then we choose an arbitrary site $x$ that has a set of influencing regions. If it has several such sets, then we make an arbitrary choice. Call this set $\mathcal{S}_i$. 
\item We suppose the entire configuration is empty except the regions in $\mathcal{S}_i$.
\item We occupy $x$. Either we now have a single occupied connected component, or we would have one after weakly enhancing. Therefore, the union of all regions in $\mathcal{S}_i$, plus $x$, forms a region.
\item $\mathcal{C}_{i+1}$ consists of all the regions in $\mathcal{C}_i$ with the regions in $S_i$ removed, and with a new region added that consists of the union of the regions in $S_i$, and $x$.  
\end{itemize}

Now the following holds: If the maximal diameter of all regions in $\mathcal{C}_i$ is $D_i$, then the maximal diameter of all regions in $\mathcal{C}_{i+1}$ is in $[D_i, (a+b+c)(D_i+2c+1)]$, and since we supposed that $[0,l]^3$ is internally spanned, the iteration does not end until the maximal diameter is $l$. Therefore, if we choose $\kappa = a+b+c$, and $\lambda = (2c+1)(a+b+c)$, the proof is complete. 
\end{proof}

\begin{definition}
In {\em enhancing} the initial configuration in a cube $[0,l]^3$, we do the following:

\begin{enumerate} 
\item We divide the cube in $l/s$ disjoint {\em slices} $[0,l]^2[i,i+s]$, with $i = 0\ldots l/s-1$. We suppose for ease of notation that $l$ is a multiple of $s$. We choose $s$ as the minimal value such that if an empty site $x$ is in slice $i$, and all other slices are fully occupied, then site $x$ still needs $a+b$ occupied sites in the intersection of its neighborhood with slice $i$, to become occupied. For instance, in \cite{cerf, cerf2} the slices have width 2, whereas for the modified model in \cite{holroydmodified} the slices can have width 1. For the $(a,b,c)$ model, we need that $a+b+c - (2c+1-s) = a+b$, this implies that we have to choose $s = c+1$. 
We say that a set of sites $(x,y,i),(x,y,i+1),\ldots,(x,y,i+s)$ is a {\em minicolumn} in slice $i$. 
For every minicolumn in every slice, if at least one site in the minicolumn is occupied, then we occupy all sites in the minicolumn. We say that $q$ is the probability for a minicolumn to be occupied, and note that $p \leq q \leq sp$, since in the initial configuration, minicolumns are independent.
\item We evolve the minicolumn configuration of every slice by the twodimensional bootstrap rule, that is, we view minicolumns as sites in the $(a,b)$ model with $q$ as parameter, and iteratively occupy those that have at least half of their neighborhood occupied. By our choice of $s$, we are not occupying sites that would not get occupied in the threedimensional model. In fact, notice that by now we have occupied at least every site that would also get occupied with the threedimensional bootstrap rule, and also a subset of the weakly spanned sites, namely the ones that are crossed by a line in the $z$-direction.
\item We fully occupy every slice where there is an occupied component of at least some size $S$. We call such a slice {\em flooded}. We will choose $S$ such that it resembles a critical droplet for the $(a,b)$ model with $q$ as parameter.
\item We occupy weakly spanned minicolumns. That is, we say a minicolumn is weakly spanned if it contains a weakly spanned site.
This step is not essential for arriving at \eqref{ecrossingbound}, but by including this step we ensure that an enhanced configuration dominates a weakly enhanced configuration, which is necessary to arrive at \eqref{finalbound}.
\item We fully occupy slice 0 and slice $L/s-1$.
\end{enumerate}

We say a rectangular block is {\em e-crossed} if it is internally crossed in the $z$-direction after enhancing. We denote $\prob_l(\mbox{e-crossed})$ for the probability that the cube $[0,l]^3$ is e-crossed. 
\label{enhanceddef}
\end{definition}

The following lemma bears resemblance to the derivation of (3.30) of \cite{cerf}, and of (23) of \cite{holroydmodified}, but we derive the bound for the $(a,b)$ model with arbitrary $a$ and $b$. In the case $(a,b) = (1,2)$, one may choose $x_c < \frac 1p \ln{\frac 1p}$ and $y_c < p^{-3/2}$ as in \cite{EH}, but note that this choice for $y_c$ is not essential: the lemma is valid for any $y_c < p^{-2+\epsilon}$. 

\begin{lemma}
For the $(a,b)$ model and $p$ small enough, the expected size $\chi$ of the occupied component of the origin, given that $x < x_c$ and $y < y_c$, with $x_c < y_c < p^{-b+\epsilon}$ for a fixed $\epsilon>0$, is at most $\sqrt{p}$. This bound does not change when weakly spanned sites are occupied.
\label{steppingstonebound}
\end{lemma}

\begin{proof}
We write 
\[
\chi \leq \sum_{x=1}^{x_c} \sum_{y=1}^{y_c} xy\prob(\mbox{ diameter} = \max\{x,y\}).
\]
For arbitrary $r$, we have 
\[ 
\chi \leq r^2 \prob(0 \leq\mbox{ diameter} \leq r) + x_c y_c \prob(r \leq \mbox{ diameter} < \max\{x_c,y_c\}).
\]
We have that if the diameter of the occupied component of the origin is at most $r$, then there must be at least one occupied site among at most $r^2$ sites. therefore, $\prob( 0 \leq \mbox{diameter} \leq r) \leq r^2 p$. If the diameter is in $[r,\max\{x_c,y_c\}]$, then we estimate $\prob(r \leq \mbox{ diameter} < \max\{x_c,y_c\}) \leq \prob(\mbox{ diameter} < \max\{x_c,y_c\})$. To estimate this probability, we generalize (2) of \cite{EH}. We first quote this equation, which was used for the case $(a,b) = (1,2)$:

\begin{equation}
\prob(\mbox{a fixed } x,y \mbox{ rectangle is internally spanned}) \leq \min\{(1-(1-\tilde{p}^2)^y)^x,(1-(1-\hat{p})^x)^y\},
\end{equation}

where $\tilde{p} \geq p$ and $\hat{p} \geq p$ are of the order $p$. In words, the meaning of this expression is that a necessary condition for a rectangle to be internally spanned, is that there is no row or column in which not even one site is spanned, even if the neighboring row or column is fully occupied. In the $y$ direction, one needs two occupied sites, close enough to each other, for a third one to be spanned. The probability for this to occur at a fixed position, is of order $p^2$. In the $x$ direction however, only one occupied site suffices.
Therefore, the generalization to arbitrary $a,b$ is straightforward:

\begin{equation}
\prob(\mbox{a fixed } x,y \mbox{ rectangle is internally spanned}) \leq \min\{(1-(1-\tilde{p}^b)^y)^x,(1-(1-\hat{p}^a)^x)^y\}.
\label{anisospanned}
\end{equation}

However, we need a bound that is also valid when weakly spanned sites are occupied as well. Note that for a site to be weakly spanned, there has to be a spanned neighbor in a nearby row/column, that is, at most $a$ resp. $b$ columns resp. rows away. Therefore, we adapt the bound as follows:

\begin{eqnarray}
\nonumber \lefteqn{\prob(\mbox{a fixed } x,y \mbox{ rectangle is internally spanned}) \leq} \\ 
\nonumber & &\min\{(1-(1-\tilde{p}^b)^{(2a+1)y})^x,(1-(1-\hat{p}^a)^{(2b+1)x})^y\}.\\
\label{anisoweaklyspanned}
\end{eqnarray}

We estimate
\begin{equation}
\begin{split}
\min\{(1-(1-\tilde{p}^b)^{(2a+1)y})^x,(1-(1-\hat{p}^a)^{(2b+1)x})^y\} & \leq (1-(1-\tilde{p}^b)^{(2a+1)y})^x \\
 & \leq (2(2a+1)y\tilde{p}^b)^x \\
 & \leq (2(2a+1)y_c\tilde{p}^b)^r.\\
 \end{split}
\end{equation}

Now we insert that $y_c < p^{-b+\epsilon}$, so that
\[
\prob(\mbox{ diameter} < \max\{x_c,y_c\}) \leq (2(2a+1)\tilde{p}^{\epsilon})^r.
\]

Therefore, we find
\[
\chi \leq r^4 p + x_c y_c(2(2a+1)\tilde{p}^\epsilon)^r \leq r^4 p + p^{-2b+2\epsilon}(2(2a+1)\tilde{p}^\epsilon)^r.
\]
We choose $r$ such that $-2b + (2+r) \epsilon > 1/2$. With this choice of $r$ and $p$ small enough, $\chi$ is less than $\sqrt{p}$, as needed.
\end{proof}

\begin{lemma}
Let $\bar{P}_{a,b} = (L^{-}_{a,b}(p))^{-1}$. 
Let $\prob_l(\mbox{e-crossed})$ and $s$ be defined as in Definition \ref{enhanceddef}.
Then
\[
\prob_l(\mbox{e-crossed}) \leq \begin{cases}
														2l^2(2\sqrt{sp})^{l/s - 2} & \mbox{ if } (1/p)^{1/4} < l < p^{-b+\epsilon},\\
														4l^2 \left(l^3 (sp)^{\frac {s-1}2} p^{2b} \bar{P}_{a,b})\right)^{l/s} & \mbox{ if } l \geq p^{-b+\epsilon}.\\
														\end{cases}
\label{enhancedlemma}
\]
\end{lemma}

\begin{proof}
Since we consider the cube $[0,l]^3$ to be e-crossed, we denote according probabilities with a subscript $l$.

We condition on the number of flooded slices:
\[
\prob_l(\mbox{e-crossed}) = \sum_{m=0}^{l/s-2} \prob_l(\mbox{e-crossed}|m \mbox{ slices flooded})\prob(m \mbox{ slices flooded}).
\]
Note that the first and last slice are fully occupied, but they are not flooded slices. By independence of the slices, we have
\[
\prob_l(m \mbox{ slices flooded}) = (\prob_l(\mbox{a slice is flooded}))^m,
\]
where the probability for a slice to be flooded is the probability that there is an occupied component of at least size $S$ for the $(a,b)$ model in the slice. 
We now condition on the positions of the flooded slices:
\[
\prob_l(\mbox{e-crossed}|m \mbox{ slices flooded}) = \sum_{i_1,\ldots,i_m}\prob_l(\mbox{e-crossed}|\mbox{slices }i_1,\ldots,i_m \mbox{ flooded}).
\]
By $\mathcal{E}_{i,j}$, for $j>i+1$, we denote the event that there exists a crossing from slice $i+1$ to slice $j-1$. Then
\[
\prob_l(\mbox{e-crossed}|\mbox{slices }i_1,\ldots,i_m \mbox{ flooded}) = \prod_{j=0}^m \prob_l(\mathcal{E}_{i_j,i_{j+1}}).
\]
The slices $i_j+1,\ldots,i_{j+1}-1$ are not flooded, therefore the largest occupied connected component is strictly smaller than $S$. If there is a crossing, then there is a self-avoiding path from slice $i_j$ to slice $i_{j+1}$ that visits a connected component in each slice in between, and possibly in some slices more than once through several disjoint connected components, like skipping between slices from one stepping stone to another. We choose the path such that the total number $H$ of connected components visited is minimal.
Let $l_j = i_{j+1} - i_j - 2$. We sum over possible values of $H$:
\[
\prob_l(\mathcal{E}_{i_j,i_{j+1}}) \leq \sum_{h \geq l_j} \prob_l(\mathcal{E}_{i_j,i_{j+1}} \mbox{in $h$ steps}).
\]
Note that $H$ can take only values $h_k = l_j+2k$, with $k = 0,1,\ldots$. Since for each larger value of $k$ the path contains an extra back- and forwardstep, for each $k$ the number of possibilities for the order of slices visited is $(l_j)^k$. We bound this by $2^{h_k}$, which is the number we get if we suppose that for every skip there are two possibilities, disregarding the fact that the path needs to lead to slice $i_{j+1}$.
We need the probability of $H$ connected components: after each stepping stone, we need a new stepping stone at the right position for the path to continue. That is, if a stepping stone is contained in a rectangle of $x$ times $y$, then there are at most $xy$ choices for the site from where we skip to the next slice. Since slice $i_j$ is fully occupied, there are $l^2$ choices for the site from which we make the first skip.
We also use that a stepping stone is smaller than $x_c y_c$, the dimensions of $S$.
We obtain
\[
\prob_l(\mathcal{E}_{i_j,i_{j+1}}) \leq l^2 \sum_{k\geq1} 2^{h_k} \prod_{i=1}^{h_k} \sum_{x_i< x_c,y_i<y_c} x_i y_i \prob_l(\mbox{occupied component has size }x_i y_i).
\]
The factor $\sum_{x_i< x_c,y_i<y_c} x_i y_i \prob_l(\mbox{occupied component has size }x_i y_i)$ does not depend on $i$, and is the expected size $\chi$ of the connected occupied component of a fixed site, given that it is smaller than $x_c y_c$.
We will choose $S$ such that we can bound this expected size by $\sqrt{q}$. By Lemma \ref{steppingstonebound}, we can choose $x_c < y_c < p^{-b+\epsilon}$ for some fixed $\epsilon>0$, so that $d_S = p^{-b+\epsilon}$, where $d_S$ is the diameter of $S$. We work out
\[
\sum_{k\geq 1} 2^{h_k} \prod_{i=1}^{h_k} \chi \leq \sum_{k\geq 1} (2\sqrt{q})^{h_k} \leq 2(2\sqrt{q})^{l_j},
\]
where the last bound is valid for $q$ small enough so that $2\sqrt{q} \leq 1/2$. 
Therefore,
\begin{equation}
\prob_l(\mathcal{E}_{i_j,i_{j+1}}) \leq 2l^2(2\sqrt{q})^{l_j}.
\label{skipping}
\end{equation}
If $S$ is a critical droplet, then we can estimate $\prob_l(\mbox{a slice is flooded}) \leq (d_S l)^2 \bar{P}_{a,b}$ if $l \geq d_S$, and 0 otherwise, since $\bar{P}_{a,b}$ is an upper bound for the probability of any choice of critical droplet (see \eqref{density-}).
However, to be sure that $S$ is a critical droplet we need $d_S \geq p^{-b}$, but we chose $d_S < p^{-b+\epsilon}$.
The probability that $S$ is occupied is therefore at most a factor $p^{-2b\epsilon}$ larger than $\bar{P}_{a,b}$, so we estimate
\[
\prob_l(\mbox{a slice is flooded}) \leq (d_S l)^2 p^{-2b\epsilon} \bar{P}_{a,b}.
\]
Putting everything together, we get that if $l \geq p^{-b+\epsilon}$, then 
\[
\prob_l(\mbox{e-crossed}) \leq \sum_{m=0}^{l/s-2} ((d_S l)^2 p^{-2b\epsilon}\bar{P}_{a,b})^m \sum_{i_1,\ldots,i_m} 4l^2 \prod_{j=0}^m q^{l_j/2}.
\]
We use that the number of possibilities for $i_1,\ldots,i_m$ is at most $l^m$, and that $l_0+l \dots+l_m = l-m$, so that $\prod_{j=0}^m q^{l_j/2} = q^{1/2(l-m)}$. Then
\[
\prob_l(\mbox{e-crossed}) \leq 4l^2 q^{l/2} \sum_{m=0}^{l/s-2} \left(\frac l{\sqrt{q}} (d_S l)^2 p^{-2b\epsilon}\bar{P}_{a,b}\right)^m,
\]
so that 
\[
\prob_l(\mbox{e-crossed}) \leq 4l^2 q^{l/2} \left(\frac l{\sqrt{q}} (d_S l)^2 p^{-2b\epsilon}\bar{P}_{a,b})\right)^{l/s}. 
\]
We insert $d_S \leq p^b$ and $q \leq sp$ and simplify, so that we finally get, for $l \geq p^{-b+\epsilon}$,
\begin{equation}
\prob_l(\mbox{e-crossed}) \leq 4l^2 \left(l^3 (sp)^{\frac {s-1}2} p^{2b(1-\epsilon)} \bar{P}_{a,b})\right)^{l/s}. 
\label{ecrossingbound}
\end{equation}
 
In the case that $l < p^{-b+\epsilon}$, we have that $\prob_l(\mbox{a slice is flooded}) = 0$. Therefore we use \eqref{skipping} with $i_j = 1$ and $i_{j+1} = l/s -1$, and obtain
\begin{equation}
\prob_l(\mbox{e-crossed}) \leq 2l^2(2\sqrt{sp})^{l/s - 2},
\label{smallecrossingbound}
\end{equation}
which tends to 0 for all $(1/p)^{1/4} < l < p^{-b+\epsilon}$.
 
\end{proof}
 
\subsection{Proof for the upper bound} 

{\it Proof of Theorem \ref{maintheorem}, upper bound for $L$}

We use Lemma \ref{spannedsubsets}
to estimate the probability of the cube $[0,L]^3$ to be internally spanned as follows:
\begin{equation}
\prob([0,L]^3 \mbox{ i.s.}) \leq L^3 \min_{1\leq k \leq \frac {L-\lambda}{\kappa}} ((\kappa-1)k+\lambda) \max_{k\leq l \leq \kappa k+\lambda} \prob_l(\mbox{e-crossed}).
\label{finalbound} 
\end{equation}

To arrive at this expression, we have reasoned as follows: If a rectangular block is weakly crossed, then the smallest cube covering this rectangular block has a crossing in some direction, after weakly enhancing. The probability for this is bounded by the probability that there is a crossing in the easiest direction after weakly enhancing. We will insert our upper bound for $\prob_l(\mbox{e-crossed})$ from Lemma \ref{enhancedlemma}. This bound was derived for the enhanced configuration, which dominates the weakly enhanced configuration. 

We consider three ranges for $L$. 

If $L < p^{-b+\epsilon}$, then we estimate
\[
\prob([0,L]^3 \mbox{ i.s.}) \leq \prob_L(\mbox{e-crossed}),
\]
which tends to 0 by \eqref{smallecrossingbound}.

Next, if $L \geq p^{-b+\epsilon}$ then we combine equations \eqref{finalbound} and \eqref{ecrossingbound}, and get
\begin{equation}
\prob([0,L]^3 \mbox{ i.s.}) \leq L^3 \min_{1\leq k \leq \frac {L-\lambda}{\kappa}} ((\kappa-1)k+\lambda) \max_{k\leq l \leq \kappa k+\lambda} 4l^2 \left(l^3 (sp)^{\frac {s-1}2} p^{2b(1-\epsilon)} \bar{P}_{a,b})\right)^{l/s}.
\label{hugemess}
\end{equation}

We will use that for $p$ small enough, $(sp)^{\frac {s-1}2} p^{2b(1-\epsilon)} <1$. 

In the case that $p^{-b+\epsilon} \leq L < (L^{-}_{a,b}(p))^{1/4}$, we choose $k = \frac {L-\lambda}\kappa$, then $(\kappa-1)k+\lambda \leq L$. 
This gives
\[
\prob([0,L]^3 \mbox{ i.s.}) \leq L^4 \max_{\frac {L-\lambda}\kappa \leq l \leq L} 4l^2 \left(l^3 \bar{P}_{a,b})\right)^{l/s}.
\]
We estimate further
\begin{eqnarray*}
\prob([0,L]^3 \mbox{ i.s.}) & \leq & 4\max_{\frac {L-\lambda}\kappa \leq l \leq L} L^6 \left(L^3 \bar{P}_{a,b})\right)^{l/s}\\
											& \leq & 4\max_{\frac {L-\lambda}\kappa \leq l \leq L} \left(L^9 \bar{P}_{a,b})\right)^{l/s}.
\end{eqnarray*}
Recall that $(L^{-}_{a,b}(p))^{1/3} = \bar{P}_{a,b})$.
If $L < (L^{-}_{a,b}(p))^{1/4}$, then $(L^9 \bar{P}_{a,b})$ tends to 0 as $p \to 0$, so that this entire expression tends to 0. 

Finally, we consider the case $L \geq (L^{-}_{a,b}(p))^{1/4}$. We now choose $k = \frac{(L^{-}_{a,b}(p))^{1/4}-\lambda}{\kappa}$ in \eqref{hugemess}, so that $\frac{(L^{-}_{a,b}(p))^{1/4}-\lambda}{\kappa} \leq l \leq (L^{-}_{a,b}(p))^{1/4}$; this means that $l = O((L^{-}_{a,b}(p))^{1/4})$. We also have $(\kappa-1)k+\lambda \leq (L^{-}_{a,b})^{1/4}$, so that:
\begin{eqnarray*}
\prob([0,L]^3 \mbox{ i.s.}) & \leq L^3  \max_l (4L^{-}_{a,b}(p))^{3/4}\left(l^3 \bar{P}_{a,b}\right)^{l/s}\\
											& \leq 4 L^3  \max_l \left((L^{-}_{a,b}(p))^{-3/2}\right)^{l/s}.\\
\end{eqnarray*}

We write $L^{-}_{a,b}(p) = e^{f(1/p)}$. 
Since $l = O((L^{-}_{a,b}(p))^{1/4})$, there is a constant $C$ such that the right factor of this expression tends to 0 as $e^{-e^{C f_{a,b}(1/p)}}$. Therefore, if $L^3$ tends to $\infty$ more slowly, then the entire expression tends to 0. This gives the required upper bound for $L$.

\qed

\paragraph{Acknowledgement:} We thank Alexander Holroyd for inspiring discussions and Hugo Duminil-Copin for a helpful correspondence.

\end{document}